\documentclass{llncs}
\usepackage{amssymb,amsmath}
\usepackage{graphics,graphicx,tabularx}
\usepackage{enumerate}
\usepackage{color}
\usepackage[compatibility=false]{caption}
\usepackage{subcaption}
\usepackage{dsfont}
\usepackage[table]{xcolor}
\usepackage{times}
\usepackage{wrapfig}

\DeclareMathOperator{\cross}{cross}
\DeclareMathOperator{\col}{c}
\DeclareMathOperator{\leftp}{left}
\spnewtheorem{numclaim}{Claim}{\itshape}{\rmfamily}

\begin{document}

\title{Many-to-One Boundary Labeling with Backbones}

\author{Michael A. Bekos\inst{1} \and Sabine Cornelsen\inst{2} \and Martin Fink\inst{3} \and Seokhee Hong\inst{4} \and Michael~Kaufmann\inst{1} \and Martin N\"ollenburg\inst{5} \and Ignaz Rutter\inst{5} \and Antonios Symvonis\inst{6}}

\institute{
    Institute for Informatics, University of  T\"ubingen, Germany\\
    \email{$\{$bekos,mk$\}$@informatik.uni-tuebingen.de}
    \and
    Department of Computer and Information Science, University of
    Konstanz, Germany\\
    \email{sabine.cornelsen@uni-konstanz.de}
    \and
    Lehrstuhl f\"ur Informatik I, Universit\"at W\"urzburg, Germany\\
    \email{martin.a.fink@uni-wuerzburg.de}
    \and
    School of Information Technologies, University of Sydney, Australia\\
    \email{shhong@it.usyd.edu.au}
    \and
    Karlsruhe Institute of Technology, Germany\\
    \email{noelle@ira.uka.de,~rutter@kit.edu}
    \and
    School of Applied Mathematics and Physical Sciences,
    NTUA, Greece\\
    \email{symvonis@math.ntua.gr}
}

\maketitle

\begin{abstract}
In this paper we study  \emph{many-to-one boundary labeling with
backbone leaders}. In this new many-to-one model, a horizontal
backbone reaches out of each label into the feature-enclosing
rectangle. Feature points that need to be connected to this label
are linked via vertical line segments to the backbone. We present
dynamic programming algorithms for label number and total leader
length minimization of crossing-free backbone labelings.  When
crossings are allowed, we aim to obtain solutions with the minimum
number of crossings. This can be achieved efficiently in the case of
fixed label order, however, in the case of  flexible label order we
show that minimizing the number of leader crossings is NP-hard.
\end{abstract}
\section{Introduction}

The process of annotating images with text in order to fully
describe specific features of interest is referred to as
\emph{labeling}. Typically, a label should not occlude features of
the image and it  should not overlap with other labels. In map
labeling, due to the small size of labels (usually a single
word/name) and our ability to control the feature density, we
usually  manage to place the labels on the map so they are in the
immediate vicinity of the feature they describe. Map labeling has
been studied in computer science for more than two
decades~\cite{FormannW1991}. Survey on algorithmic map labeling and
an extensive bibliography are given by Neyer~\cite{Neyer2001} and
Wolff and Strijk~\cite{WS96}, respectively. However, \emph{internal
labeling}  is not feasible when large labels are employed, a typical
situation in technical drawings and medical atlases. \emph{Boundary
labeling} was developed by Bekos et al.~\cite{BekosKSW2007} as a
framework and an algorithmic response to the poor quality (feature
occlusion, label overlap) of specific labeling applications. In
boundary labeling,  labels are placed at the boundary of a rectangle
and are connected to their associated features via arcs referred to
as \emph {leaders}. Leaders attach to labels at \emph{label-ports}.
A survey article  by Kaufmann~\cite{Kau09} presents the different
boundary labeling models that have been studied in the literature.

As in map labeling, most work  on boundary labeling has been devoted
to the case where each label is associated with a single feature
point. However, the case where each label is associated to more than
one feature point (the topic of this paper) is also common in
applications. We can think of groups of features sharing common
properties (e.g., identical components of technical devices or
locations of plants/animals of the same species in a map), which we
express as having the same color. Then, we need to connect via
leaders these identically colored feature points to a label of the
same color. \emph{Many-to-one boundary labeling} was formally
introduced by Lin et al.~\cite{LinKY08}. In their initial definition
of many-to-one labeling each label had one port for each connecting
feature point, i.e., each point uses an individual leader (see
Fig.~\ref{fig:mto_orig}). This inevitably lead to (i) tall labels
(ii) a wide track-routing area between the labels and the enclosing
rectangle (since leaders are not allowed to overlap) and (iii)
leader crossings in the track routing area. Lin et
al.~\cite{LinKY08} examined one and two-sided boundary labeling
using so-called $opo$-leaders~\cite{BekosKSW2007}. They showed that
several crossing minimization problems are NP-complete and,
subsequently, developed approximation and heuristic algorithms. In a
variant of this model, referred to as \emph{boundary labeling with
hyperleaders}, Lin~\cite{Lin2010} resolved the multiple port issue
by joining together all leaders attached to a common label with a
vertical line segment in the track-routing area (see
Fig.~\ref{fig:mto_hyper}). At the cost of label duplications, leader
crossings could be eliminated.

\begin{figure}[t]
    \begin{subfigure}[b]{.32\textwidth}
        \centering
        \includegraphics[width=\textwidth]{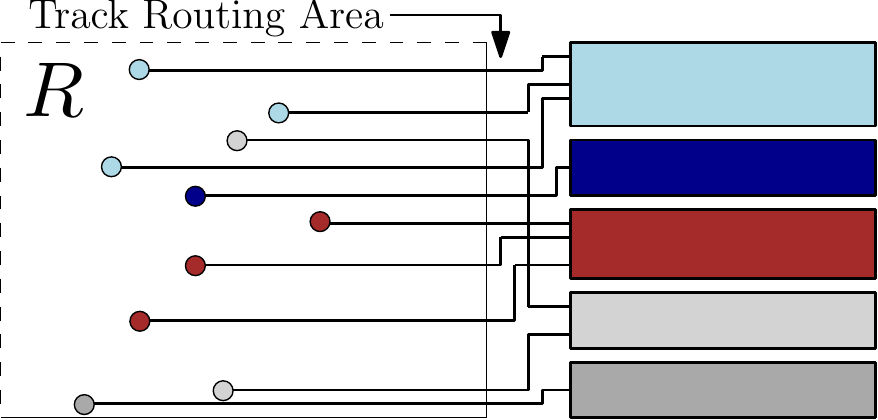}
        \caption{Individual leaders~\cite{LinKY08}}
        \label{fig:mto_orig}
    \end{subfigure}
    \hfill
    \begin{subfigure}[b]{.32\textwidth}
        \centering
        \includegraphics[width=\textwidth]{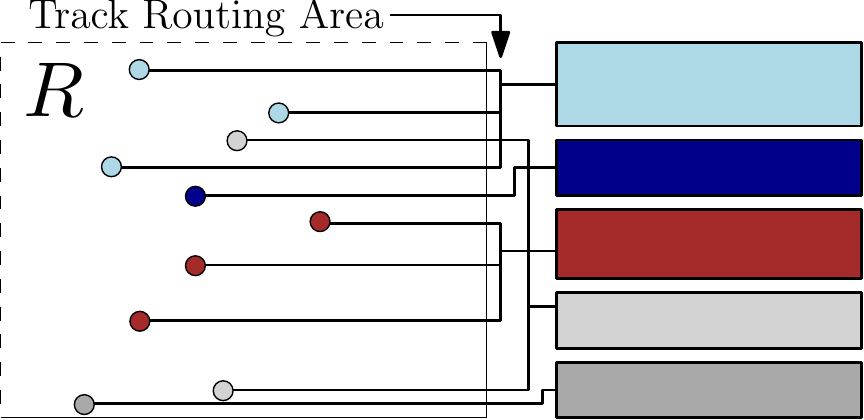}
        \caption{Hyperleaders~\cite{Lin2010}}
        \label{fig:mto_hyper}
    \end{subfigure}
    \hfill
    \begin{subfigure}[b]{.32\textwidth}
        \centering
        \includegraphics[width=\textwidth]{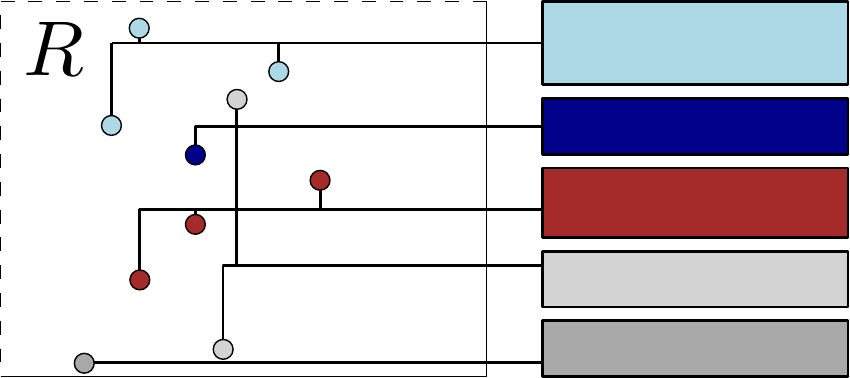}
        \caption{Backbones}
        \label{fig:mto_backbone}
    \end{subfigure}
\caption{Different types of many-to-one labelings.}
\label{fig:examples}
\end{figure}

We study \emph{many-to-one boundary labeling with backbone leaders}
(for short, \emph{backbone labeling}). In this many-to-one model, a
horizontal backbone reaches out of each label into the
feature-enclosing rectangle. Feature points that need to be
connected to a label are linked via vertical line segments to the
label's backbone (see Fig.~\ref{fig:mto_backbone}). The backbone
model does not need a  track routing area and thus overcomes several
disadvantages of previous many-to-one labeling models, in particular
the issues (ii) and (iii) mentioned above. As
Fig.~\ref{fig:examples} shows, backbone  labelings also require much
less ``ink'' in the image than the previous methods and thus is
expected to be less disturbing for the viewer. We note that backbone
labeling can be seen as a variation of Lin's $opo$-hyperleaders.
Lin~\cite{Lin2010} posed it as an open problem to study
$po$-hyperleaders (which is his terminology for backbones), in
particular to minimize the number of duplicate labels in a
crossing-free labeling.

\subsection{Our Contribution}

Our paper studies three aspects of backbone  labeling, \emph{label
number minimization} (Section~\ref{sec:number}), \emph{total leader
length minimization} (Section~\ref{sec:length}), and \emph{crossing
minimization} (Section~\ref{sec:cross}). The first two aspects
require crossing-free leaders. We consider both \emph{finite
backbones} and \emph{infinite backbones}. Finite backbones extend
horizontally from the label to the furthest point connected to the
backbone, whereas infinite backbones span the whole width of the
rectangle (thus one could use duplicate labels on both sides).
Furthermore, our algorithms vary depending on whether the order of
the labels is fixed or flexible and whether more than one label per
color class can be used.

For crossing-free backbone labeling we derive efficient algorithms
based on dynamic programming to minimize label number and total
leader length (Sections~\ref{sec:number} and~\ref{sec:length}),
which solves the open problem of Lin~\cite{Lin2010}. The main idea
is that backbones can be used to split an instance into two
independent subinstances. For infinite leaders faster algorithms are
possible since each backbone generates two independent instances;
for finite backbones the algorithms require more effort since a
backbone does not split the whole point set and thus the outermost
point connected to each backbone must be considered. If crossings
are allowed, we give an efficient algorithm for crossing
minimization with fixed label order and show NP-completeness for
flexible label order (Section~\ref{sec:cross}).

\subsection{Problem Definition}
In backbone labeling, we are given  a set $P$ of $n$ points in an
axis-aligned rectangle $R$, where each point $p\in P$ is assigned a
color $c(p)$ from a color-set $C$.  Our goal is to  place colored
labels on the boundary of $R$ and to assign each point $p \in P$ to
a distinct label $l(p)$ of color $c(p)$.

All points assigned to the same label will be connected to it
through a single backbone leader. A \emph{backbone leader} consists
of a horizontal \emph{backbone}  attached to the left or right side
of the enclosing rectangle $R$  and vertical line segments that
connect the points to the backbone.
Since the backbones are horizontal, we consider labels to be fully
described by the y-coordinate of their backbone. Note that, at first
sight, this may imply that labels are of infinitely small height.
However, by imposing a minimum separation distance between
backbones, we can also accommodate labels of finite height.

Let $\mathcal{L}$ be a set of colored labels and consider label $l
\in \mathcal{L}$. By $c(l)$, $y(l)$, and $P(l)$ we denote the color
of label $l$,
the $y$-coordinate of the backbone of label $l$ on the boundary
of $R$ and the set of points that are connected/associated to label
$l$, respectively.

A \emph{backbone (boundary) labeling} for a set of colored points
$P$  in a rectangle $R$ is a set $\mathcal{L}$ of colored labels
together with a mapping of each point $p \in P$  to some
$c(p)$-colored label  in $\mathcal{L}$. The drawing can be easily
produced  since  the backbone  leader for label $l$ is fully
specified  by  $y(l)$ and $P(l)$. A backbone  labeling is called
\emph{legal} if and only if (i) each point is connected to a label
of the same color, and (ii) there are no backbone leader overlaps
(though crossings are allowed in some cases).

Several restrictions on the number of labels of a specific color may
be imposed: The number of labels may be unlimited, effectively
allowing us to assign each point to a distinct label. Alternatively,
the number of labels  may be bounded by $K \geq |C|$. If $K=|C|$,
all points of the same color have to be assigned to a single label.
We may also restrict the maximum number of allowed labels for each
color in $C$ separately by specifying  a \emph{color vector}
$\vec{k} = (k_{1}, \ldots, k_{|C|})$. A legal backbone  labeling
that satisfies all of the imposed restrictions on the number of
labels is called \emph{feasible}. Our goal in this paper is to find
feasible backbone labelings that optimize different quality
criteria.

A backbone  labeling without leader crossings is referred to as
\emph{crossing-free}. An interesting variation of the backbone
labeling concerns the size of the backbone. A \emph{finite backbone}
attached to a label at, say, the right side of $R$ extends up to the
leftmost point that is assigned to it. An \emph{infinite backbone}
spans the whole width of $R$.  Note that, in the case of
crossing-free labelings,  infinite backbones may result in labelings
with a larger number of labels and increased total leader length.

In the rest of the paper, we denote the points of  $P$ as $\{ p_1,
p_2, \ldots, p_n\}$ and we assume that no two points share the same
x or y-coordinate. For simplicity,  we consider the points to be
sorted in decreasing order of y-coordinates, with $p_1$ being the
topmost point in all of our relevant drawings.

\section{Minimizing the Total Number of Labels}\label{sec:number}
In this section we minimize the total number of labels in a
crossing-free solution, i.e., we set $K=n$ so that there is
effectively no upper bound on the number of labels.

\subsection{Infinite Backbones.}

We start with an important observation on the structure of
crossing-free labelings with infinite backbones.
\begin{lemma}
Let $p_i, p_{i+1}$ be two points that are vertically consecutive.
Let $p_j$ ($j<i$) be the first point above $p_i$ with $\col(p_j)
\neq \col(p_i)$, and let $p_{j'}$ ($j' > i+1$) be the first point
below $p_{i+1}$ with $\col(p_{j'}) \neq \col(p_{i+1})$ if such
points exist.

In any crossing-free backbone labeling $p_i$ and $p_{i+1}$ are
vertically separated by at most 2 backbones. Furthermore, any
separating backbone has color $\col(p_i), \col(p_{i+1})$,
$\col(p_{j})$, or~$\col(p_{j'})$. \label{lemma:intermediate-colors}
\end{lemma}
\begin{proof}
Suppose there are three separating backbones. Then the middle one
could not be connected to any point. Now, suppose a separating
backbone is connected to a point $p_{k}$ above $p_i$ and has color
$\col(p_k) \notin \{\col(p_j), \col(p_i)\}$. Then $k < j < i$. The
backbone for $p_j$ has to be above $p_k$. Point $p_i$ is lying
between two backbones of other colors; see
Fig.~\ref{fig:intermediate-colors-lemma}. Its own backbone cannot be
placed there without crossing a vertical segment connecting $p_k$ or
$p_j$ to their corresponding backbone. Symmetrically, we see that a
backbone separating $p_i$ and $p_{i+1}$ that is connected to a point
below $p_{i+1}$ can only have color $\col(p_{i+1})$ or
$\col(p_{j'})$.\qed
\end{proof}

\begin{figure}[ht]
    \centering
    \includegraphics{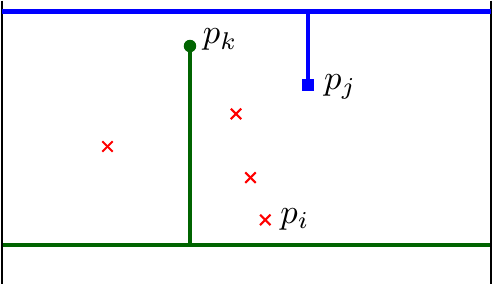}
    \caption{Point $p_i$ cannot be labeled.}
    \label{fig:intermediate-colors-lemma}
\end{figure}

Clearly, if all points have the same color, one label always
suffices. Even in an instance with two colors, one label per color
is enough: We place the backbone of one color above all points, and
the backbone of the second color below all points. However, if a
third color is involved, then many labels may be required.

We denote by $NL(P)$ the number of labels of an optimal
crossing-free solution of $P$.  In the general case of the problem,
$P$ may contain several consecutive points of the same color. We
proceed to construct a new instance $C(P)$ based on instance $P$, in
which no two consecutive points are of the same color. To do so, we
identify each maximal set of same-colored consecutive points of $P$
and we replace them by a single point of the same color that lies in
the position of the topmost point of this set. Note that in order to
achieve this, a simple top-to-bottom sweep is enough. Let
$C(P)=\{p_1',p_2',\ldots,p_k'\}$ be the \emph{clustered point set},
that we just constructed. For the sake of simplicity, we assume that
$f\colon P \rightarrow C(P)$ is a function which computes the
representative in $C(P)$ for a given a point of $P$, and, with a
slight abuse of notation let $f^{-1}\colon
P\rightarrow\mathcal{P}(C(P))$ be its ``inverse function''.

\begin{lemma}
The number of labels needed in an optimal crossing-free labeling of
$P$ with infinite backbones is equal to the number of labels needed
in an optimal crossing-free solution of $C(P)$, i.e.,
$NL(P)=NL(C(P))$. \label{lemma:clustered_point_set}
\end{lemma}
\begin{proof}
Since $C(P) \subseteq P$, it trivially follows that $NL(C(P)) \leq
NL(P)$. So, in order to complete the proof it remains to show that
$NL(P) \leq NL(C(P))$. Let $S(C(P))$ be an optimal solution of
$C(P)$ with $NL(C(P))$ labels. If we manage to construct a solution
of $P$ that has exactly the same number of labels as the optimal
solution of $C(P)$, then obviously $NL(P) \leq NL(C(P))$.

Let $p_{i}'$, $i=1,2,\ldots,k$, be an arbitrary point of $C(P)$ and
let $\{p_j,p_{j+1},\ldots,p_{j+m}\}$ be the maximal set of
consecutive, same-colored points of $P$  that has $p_{i}'$ as its
representative in $C(P)$, i.e.,
$f^{-1}(p_{i}')=\{p_j,p_{j+1},\ldots,p_{j+m}\}$. Let $S(p_{i}')$ be
the horizontal strip that is defined by the two horizontal lines
through $p_j$ and $p_{j+m}$, respectively. Clearly, in a legal
solution of $P$, $S(p_{i}')$ can accommodate at most one backbone,
i.e., the one of $\{p_j,p_{j+1},\ldots,p_{j+m}\}$, as we look for
crossing-free solutions. Now, observe that
$S(p_{1}'),S(p_{2}'),\ldots,S(p_{k}')$ do not overlap with each
other, since we have assumed that our point set $P$ is in general
position, and subsequently, all maximal sets of consecutive,
same-colored points of $P$ are well separated. We proceed to derive
a first solution $S(P)$ of $P$ from $S(C(P))$ as follows: We connect
each point $p_i$ to the backbone of its representative $f(p_i)$ in
$S(C(P))$. Clearly, $S(P)$ is not necessarily crossing-free.
However, all potential crossings should appear in horizontal strips
$S(p_{1}'),S(p_{2}'),\ldots,S(p_{k}')$; otherwise $S(C(P))$ is not
crossing-free as well. Let $S(p_{i}')$, $i=1,2,\ldots,k$, be a
horizontal strip which contains crossings. As already stated,
$S(p_{i}')$ can accommodate at most one backbone, i.e., the one of
$\{p_j,p_{j+1},\ldots,p_{j+m}\}$. We proceed to move all backbones
in $S(p_{i}')$ that are above (below, resp.) the one of
$\{p_j,p_{j+1},\ldots,p_{j+m}\}$ on to the top of (below, resp.)
$S(p_{i}')$, without changing their relative order and without
influencing the strips above and below $S(p_{i}')$ (recall that
$S(p_{1}'),S(p_{2}'),\ldots,S(p_{k}')$ do not overlap with each
other, which suggests that this is always possible). From the above
it follows that the constructed solution is crossing-free and has
the same number of labels as the one of $C(P)$, which completes the
proof of this lemma.\qed
\end{proof}

We now present a linear-time algorithm for minimizing the number of
infinite backbones.

\begin{theorem}
Let $P=\{p_1,p_2,\ldots,p_n\}$ be an input point set consisting of
$n$ points sorted from top to bottom. Then, a crossing-free labeling
of $P$ with the minimum number of infinite backbones can be computed
in $O(n)$ time. \label{theorem:infinite-min-number}
\end{theorem}
\begin{proof}
In order to simplify the proof, we assume that no two consecutive
points have the same color, with the aid of
Lemma~\ref{lemma:clustered_point_set}. For $i=1,2,\ldots,n$,
$\{c_{\mathrm{bak}},c_{\mathrm{free}}\} \subseteq C$, and,
$\mathrm{cur} \in \{true,false\}$, let $nl(i, \mathrm{cur},
c_{\mathrm{bak}}, c_{\mathrm{free}})$ be the optimal number of
backbones above or at $p_i$
in the case where:

\begin{itemize}
  \item The lowest backbone has color $c_{\mathrm{bak}}$.
  \item If $\mathrm{cur}$ is true, the lowest backbone coincides with $p_i$.
  Hence, it is $c(p_i)$-colored, i.e., $c_{\mathrm{bak}} = c(p_i)$;
  otherwise the lowest backbone is above $p_i$.
  Note that in the latter case $p_i$ might be unlabeled (e.g., the
  color of the lowest backbone is not $c(p_i)$, or equivalently, $c_{\mathrm{bak}} \neq c(p_i)$).
  \item The point that, by Lemma~\ref{lemma:intermediate-colors}, may
  exist between $p_i$ and the lowest backbone has color $c_{\mathrm{free}}$.
  Obviously, in the case where $\mathrm{cur}=\mathrm{true}$ (i.e., the
  lowest backbone
  coincides with $p_i$) this point does not exist.
  So, in general, if this point does not exists, we assume
  that $c_{\mathrm{free}}=\emptyset$.
\end{itemize}

Obviously, $nl(1, true, c(p_1), \emptyset)=1$ and $nl(1, false,
\emptyset, c(p_1))=0$. Now assume that we have computed all entries
of table $nl$ that correspond to different labelings induced by
point $p_i$. In order to compute the corresponding table entries for
the next point $p_{i+1}$, we distinguish two cases:

\begin{figure}[t]
    \centering
    \fbox{\begin{subfigure}[b]{.14\textwidth}
        \centering
        \includegraphics[width=\linewidth]{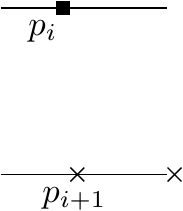}
        \caption{}
        \label{fig:1}
    \end{subfigure}}
    \hfill
    \fbox{\begin{subfigure}[b]{.14\textwidth}
        \centering
        \includegraphics[width=\linewidth]{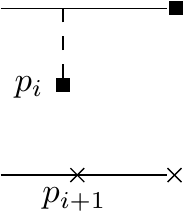}
        \caption{}
        \label{fig:2}
    \end{subfigure}}
    \hfill
    \fbox{\begin{subfigure}[b]{.14\textwidth}
        \centering
        \includegraphics[width=\linewidth]{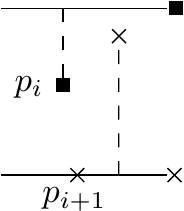}
        \caption{}
        \label{fig:3}
    \end{subfigure}}
    \hfill
    \fbox{\begin{subfigure}[b]{.14\textwidth}
        \centering
        \includegraphics[width=\linewidth]{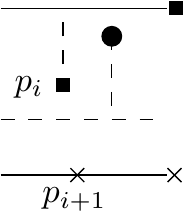}
        \caption{}
        \label{fig:4}
    \end{subfigure}}
    \hfill
    \fbox{\begin{subfigure}[b]{.14\textwidth}
        \centering
        \includegraphics[width=\linewidth]{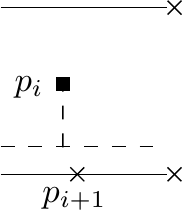}
        \caption{}
        \label{fig:5}
    \end{subfigure}}
    \hfill
    \fbox{\begin{subfigure}[b]{.14\textwidth}
        \centering
        \includegraphics[width=\linewidth]{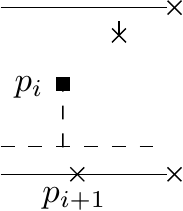}
        \caption{}
        \label{fig:6}
    \end{subfigure}}
    \hfill
    \fbox{\begin{subfigure}[b]{.14\textwidth}
        \centering
        \includegraphics[width=\linewidth]{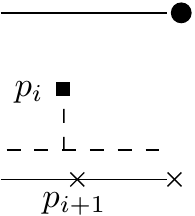}
        \caption{}
        \label{fig:7}
    \end{subfigure}}
    \hfill
    \fbox{\begin{subfigure}[b]{.14\textwidth}
        \centering
        \includegraphics[width=\linewidth]{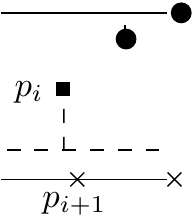}
        \caption{}
        \label{fig:8}
    \end{subfigure}}
    \hfill
    \fbox{\begin{subfigure}[b]{.14\textwidth}
        \centering
        \includegraphics[width=\linewidth]{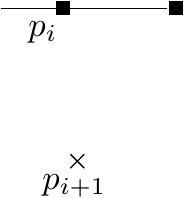}
        \caption{}
        \label{fig:9}
    \end{subfigure}}
    \hfill
    \fbox{\begin{subfigure}[b]{.14\textwidth}
        \centering
        \includegraphics[width=\linewidth]{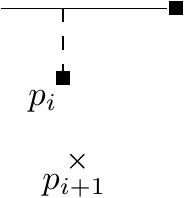}
        \caption{}
        \label{fig:10}
    \end{subfigure}}
    \hfill
    \fbox{\begin{subfigure}[b]{.14\textwidth}
        \centering
        \includegraphics[width=\linewidth]{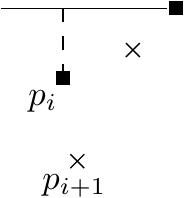}
        \caption{}
        \label{fig:11}
    \end{subfigure}}
    \hfill
    \fbox{\begin{subfigure}[b]{.14\textwidth}
        \centering
        \includegraphics[width=\linewidth]{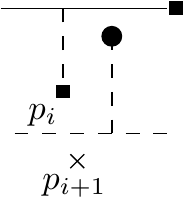}
        \caption{}
        \label{fig:12}
    \end{subfigure}}
    \hfill
    \fbox{\begin{subfigure}[b]{.14\textwidth}
        \centering
        \includegraphics[width=\linewidth]{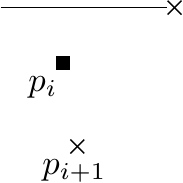}
        \caption{}
        \label{fig:13}
    \end{subfigure}}
    \hfill
    \fbox{\begin{subfigure}[b]{.14\textwidth}
        \centering
        \includegraphics[width=\linewidth]{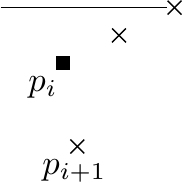}
        \caption{}
        \label{fig:14}
    \end{subfigure}}
    \hfill
    \fbox{\begin{subfigure}[b]{.14\textwidth}
        \centering
        \includegraphics[width=\linewidth]{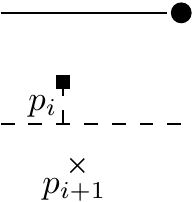}
        \caption{}
        \label{fig:15}
    \end{subfigure}}
    \hfill
    \fbox{\begin{subfigure}[b]{.14\textwidth}
        \centering
        \includegraphics[width=\linewidth]{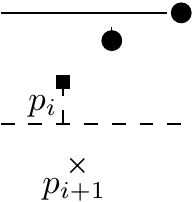}
        \caption{}
        \label{fig:16}
    \end{subfigure}}
    \hfill
    \fbox{\begin{subfigure}[b]{.14\textwidth}
        \centering
        \includegraphics[width=\linewidth]{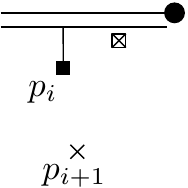}
        \caption{}
        \label{fig:17}
    \end{subfigure}}
    \hfill
    \fbox{\begin{subfigure}[b]{.14\textwidth}
        \centering
        \includegraphics[width=\linewidth]{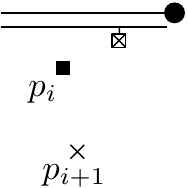}
        \caption{}
        \label{fig:18}
    \end{subfigure}}
    \caption{Different configurations that arise in proof of Theorem~\ref{theorem:infinite-min-number}.}
    \label{fig:cases}
\end{figure}

\begin{enumerate}
\item \label{case:on} \emph{The lowest backbone coincides with $p_{i+1}$:}
In this case, the lowest backbone should be $c(p_{i+1})$-colored,
$\mathrm{cur}=true$ and obviously there is no unlabeled point
between the backbone through $p_{i+1}$ and $p_{i+1}$, i.e.,
$c_{\mathrm{free}}=\emptyset$. Hence, we need to compute entry
$nl({i+1}, true, c(p_{i+1}), \emptyset)$. To do so, we distinguish
the following subcases with respect to the color of the lowest
backbone, say $bb$, above or at point $p_i$:

\begin{enumerate}[\ref{case:on}.1]
\item \emph{$bb$ is above or at point $p_i$ and $c(p_i)$-colored}.
If $bb$ is at point $p_i$ (see Fig.~\ref{fig:1}), then trivially
there is no unlabeled point below it. Hence, a feasible solution can
be derived from $nl(i, true, c(p_i), \emptyset)$ by adding a new
backbone, i.e., the one incident to $p_{i+1}$. If $bb$ is above
point $p_i$, then we distinguish two subcases. (a)~If there is no
unlabeled point below it (see Fig.~\ref{fig:2}), a feasible solution
can be derived from $nl(i, false, c(p_i), \emptyset)$ again by
adding a new backbone, i.e., the one incident to $p_{i+1}$. (b)~On
the other hand, if there is an unlabeled point below $bb$, then
again we need to distinguish two subcases. (b.1) If the unlabeled
point is colored $c(p_{i+1})$ (see Fig.~\ref{fig:3}), then a single
backbone, i.e., the one incident to $p_{i+1}$ is enough. The
corresponding solution is derived from $nl(i, false, c(p_i),
c(p_{i+1}))$. (b.2) However, in the case where the unlabeled point
is colored $c$ and $c \notin \{c(p_{i}),c(p_{i+1})\}$ (see
Fig.~\ref{fig:4}), two backbones are required and the corresponding
feasible solution is derived from $nl(i, false, c(p_i), c)$, $c
\notin \{c(p_{i}),c(p_{i+1})\}$. Note that the case where the
unlabeled point below $bb$ is of color $c(p_i)$ cannot occur, since
we have assumed that consecutive points are not of the same color.

\item \emph{$bb$ is above $p_i$ and $c(p_{i+1})$-colored}.
Again, we distinguish two subcases. (a)~If there is no unlabeled
point below it (see Fig.~\ref{fig:5}), then a feasible solution can
be derived from $nl(i, false, c(p_{i+1}), \emptyset)$ by adding two
new backbones, i.e., the one incident to $p_i$ and the one incident
to $p_{i+1}$. (b)~If there is an unlabeled point below it (see
Fig.~\ref{fig:6}), then its color should be $c(p_{i+1})$. If this is
not the case, it is easy to see that the backbone above $p_i$ is not
$c(p_{i+1})$-colored. Again two backbones are required, i.e., the
one incident to $p_i$ and the one incident to $p_{i+1}$. The
corresponding solution is derived from $nl(i, false, c(p_{i+1}),
c(p_{i+1}))$.

\item \emph{$bb$ is above $p_i$ and $c$-colored, where $c \neq c(p_i)$ and $c
\neq c(p_{i+1})$}. In this case, either there is no unlabeled point
below $bb$ (see Fig.~\ref{fig:7}) or there is one which is
$c$-colored (see Fig.~\ref{fig:8}). In both cases, two backbones
have to be placed; one incident to $p_i$ and one incident to
$p_{i+1}$. In the former case, the corresponding feasible solution
is derived from $nl(i, false, c, \emptyset)$, $c \notin
\{c(p_{i}),c(p_{i+1})\}$, while in the latter it is derived from
$nl(i, false, c, c)$, $c \notin \{c(p_{i}),c(p_{i+1})\}$.
\end{enumerate}

From the above cases, it follows:

\begin{displaymath}
nl({i+1}, true, c(p_{i+1}), \emptyset) = \min \left\{
\begin{array}{lr}
nl(i, true, c(p_i), \emptyset) + 1 \\
nl(i, false, c(p_i), \emptyset) + 1 \\
nl(i, false, c(p_i), c(p_{i+1})) + 1\\
nl(i, false, c(p_i), c) + 2, c\notin \{c(p_{i}),c(p_{i+1})\}\\
nl(i, false, c(p_{i+1}), \emptyset) + 2 \\
nl(i, false, c(p_{i+1}), c(p_{i+1})) + 2 \\
nl(i, false, c, \emptyset)+2, c \notin \{c(p_{i}),c(p_{i+1})\} \\
nl(i, false, c, c)+2, c \notin \{c(p_{i}),c(p_{i+1})\} \\
\end{array}
\right.
\end{displaymath}

\item \label{case:above} \emph{The lowest backbone is above $p_{i+1}$:}
Again, we distinguish subcases with respect to the color of the
lowest backbone $bb$ above or at point $p_i$:

\begin{enumerate}[\ref{case:above}.1]
\item \emph{$bb$ is above or at point $p_i$ and $c(p_i)$-colored}.
If $bb$ is at point $p_i$ (see Fig.~\ref{fig:9}) or $bb$ is above
point $p_i$ and either there is no unlabeled point below it (see
Fig.~\ref{fig:10}) or the unlabeled point below it is colored
$c(p_{i+1})$ (see Fig.~\ref{fig:11}), then no backbone is required.
Then, the corresponding feasible solutions are as follows:
$$nl({i+1}, true, c(p_i), \emptyset) = \min\{nl(i, true, c(p_i), \emptyset), nl(i, false, c(p_i), \emptyset)\}$$
$$nl({i+1}, true, c(p_i), c(p_{i+1})) = nl(i, false, c(p_i), c(p_{i+1}))$$
However, in the case where the unlabeled point below $bb$ is
$c$-colored, where $c \neq c(p_i)$ and $c \neq c(p_{i+1})$, then a
backbone is required (see Fig.~\ref{fig:12}). Hence, the
corresponding feasible solution can be derived as follows:
\[nl({i+1}, true, c(p_i), c) = nl(i, false, c(p_i),
c)+1,~c\notin\{c(p_i), c(p_{i+1})\}\]

\item \emph{$bb$ is above $p_i$ and $c(p_{i+1})$-colored}.
In this case, either there is no unlabeled point below it (see
Fig.~\ref{fig:13}) or there is one which is $c(p_{i+1})$-colored
(see Fig.~\ref{fig:14}). In both cases no backbone is required.
Hence, the corresponding feasible solutions can be derived as
follows:
$$nl({i+1}, false, c(p_{i+1}), \emptyset) = nl(i, false, c(p_{i+1}), \emptyset)$$
$$nl({i+1}, false, c(p_{i+1}), c(p_{i+1})) = nl(i, false, c(p_{i+1}), c(p_{i+1}))$$

\item \emph{$bb$ is above $p_i$ and $c$-colored, where $c \neq c(p_i)$ and $c
\neq c(p_{i+1})$}. In this case, if there is no unlabeled point
below it (see Fig.~\ref{fig:15}) or there is one which is
$c$-colored (see Fig.~\ref{fig:16}), then one backbone is required
for $p_i$. The corresponding feasible solution can be derived as
follows:
\begin{align*}
  nl({i+1}, false, c(p_i), \emptyset) = \min&\{nl(i, false, c,
\emptyset),\\ &nl(i, false, c, c)\}+1,~c\notin\{c(p_i),c(p_{i+1})\}
\end{align*}
The most interesting case of our case analysis arises when a forth
color is involved, say $c'\notin\{c(p_i),c(p_{i+1}),c\}$. In this
case, either the $c'$-colored point remains unlabeled and $p_i$ is
labeled (see Fig.~\ref{fig:17}), or, the $c'$-colored point is
labeled and $p_i$ remains unlabeled (see Fig.~\ref{fig:18}). The
corresponding feasible solutions can be described as follows.
\begin{align*}
  nl(&{i+1}, false, c(p_i), c')\\ &= nl(i, false, c,
  c')+1,~c\notin\{c(p_i),c(p_{i+1})\}, c' \notin\{c(p_i),c(p_{i+1}),c\}
\end{align*}
\begin{align*}
  nl(&{i+1}, false, c', c(p_i))\\ &= nl(i, false, c,
  c')+1,~c\notin\{c(p_i),c(p_{i+1})\}, c' \notin\{c(p_i),c(p_{i+1}),c\}
\end{align*}

\end{enumerate}
\end{enumerate}
Having computed table $nl$, the number of labels of the optimal
solution of $P$ equals to the minimum entry of the form
$nl(n,false,\cdot,\emptyset)$. Since the algorithm maintains an $n
\times 2 \times |C| \times |C|$ table and each entry is computed in
constant time, the time complexity of our algorithm is $O(n|C|^2)$.
However, with the aid of Lemma~\ref{lemma:clustered_point_set}, it
can be reduced to $O(n)$ (since there is a constant number of colors
that ``surround'' each point).  A corresponding solution can be
found by backtracking in the dynamic program. \qed
\end{proof}

\subsection{Infinite Backbones.}
\label{sec:finite-backbones-label-min}
We now consider finite backbones. First, note that we can always
slightly shift the backbones in a given solution so that backbones
are placed only in gaps between points. We number the gaps from $0$
to $n$ where gap $0$ is above and gap $n$ is below all the points,
respectively.

Suppose a point $p_l$ lies between a backbone of color $c$ in gap
$g$ and a backbone of color $c'$ in gap $g'$ with $0\le g < l \le g'
\le n$ such that both backbones horizontally extend to at least the
x-coordinate of $p_l$; see Fig.~\ref{fig:finite-backbones-dp}.
Suppose all points except the ones in the rectangle $R(g,g',l)$,
spanned by the gaps $g$ and $g'$ and limited by $p_l$ to the left
and by the boundary to the right, are already labeled. An optimum
solution for connecting the points in $R$ cannot reuse any backbone
except for the two backbones in gaps $g$ and $g'$; hence, it is
independent of the rest of the solution.

We use this observation for solving the problem by a dynamic
program. For $0 \le g \le g' \le n$, $l \in \left\{ g, \ldots ,g'
\right\} \cup \left\{ \emptyset \right\}$, and two colors $c$ and
$c'$ let $T[g,c,g',c',l]$ be the minimum number of additional labels
that are needed for labeling all points in the rectangle $R(g,g',l)$
under the assumption that there is a backbone of color $c$ in gap
$g$, a backbone of color $c'$ in gap $g'$, between these two
backbones there is no backbone placed yet, and they both extend to
the left  of $p_l$. Note that for $l=\emptyset$ the rectangle is
empty and $T[g,c,g',c',\emptyset] = 0$.

\begin{figure}[htb]
    \centering
    \includegraphics[scale=0.94]{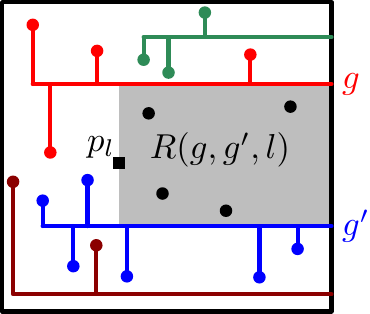}
    \caption{A partial instance bounded by two backbones and the leftmost point $p_l$.}
    \label{fig:finite-backbones-dp}
\end{figure}

We distinguish cases based on the connection of $p_l$. First, if
$c(p_l)=c$ or $c(p_l) = c'$, it is always optimal to connect $p_l$
to the top or bottom backbone, respectively, as all remaining points
will be to the right of the new vertical segment. Hence, in this
case, $T[g,c,g',c',l] = T[g,c,g',c',\leftp(g,g',l)]$ where
$\leftp(g,g',l)$ is the index of the leftmost point in the interior
of $R(g,g',l)$ or $\leftp(g,g',l) = \emptyset$ if no such point
exists.

Otherwise suppose $c(p_l) \notin \{c,c'\}$. For connecting
$p_l$ we need to place a new backbone of color $c(p_l)$, which is
possible in any gap $\tilde{g}$ with $g\le\tilde{g}\le g'$. The
backbone splits the instance into two parts, between gaps $g$ and
$\tilde{g}$ and between gaps $\tilde{g}$ and $g'$. Hence, we obtain
the recursion $T[g,c,g',c',l] = 1+ \min_{g \le \tilde{g} \le g'}
\left(
   T[g,c,\tilde{g},c(p_l),\leftp(g,\tilde{g},l)]\right.
  +$ $\left. T[\tilde{g},
c(p_l), g',c', \leftp(\tilde{g},g',l)] \right)
$.

Finally, let $\bar{c} \notin C$  be a dummy color, and let
$p_{\bar{l}}$ be the leftmost point. Then the value
$T[0,\bar{c},n,\bar{c},\bar{l}]$ is the minimum number of labels
needed for labeling all points. We can compute each of the $(n+1)
\times |C| \times (n+1) \times |C| \times (n+1)$ entries of table
$T$ in $O(n)$ time. All $\leftp(\cdot,\cdot,\cdot)$-values can
easily be precomputed  in $O(n^3)$ total time.

\begin{theorem}
Given a set $P$ of $n$ colored points and a color set $C$, we can
compute a feasible labeling of $P$ with the minimum number of finite
backbones in $O(n^4 |C|^2)$ time.
\label{thm:finite-label-min}
\end{theorem}
\section{Length Minimization}\label{sec:length}
In this section we minimize the total length of all leaders in a
crossing-free solution, either including or excluding the horizontal
lengths of the backbones. We distinguish between a global bound $K$
on the number of labels or a color vector $\vec{k}$ of individual
bounds per color.

\subsection{Infinite Backbones.}
\label{sec:infinite-lenth-minimization}
We use a parameter $\lambda$ to distinguish the two minimization
goals, i.e., we set $\lambda = 0$, if we want to minimize only the
sum of the length of all vertical segments and we set $\lambda$ to
be the width of the rectangle $R$ if we also take the length of the
backbones into account.  In this section, we assume that
$p_1>\dots>p_n$ are the y-coordinates of the input points.

\paragraph{Single Color.}
If all points have the same color, we have to choose a set $S$ of at
most $K$ y-coordinates where we draw the backbones and connect each
point to its nearest backbone.  Hence, we must solve the following
problem: Given $n$ points with y-coordinates $p_1>\dots>p_n$, find a
set $S$ of at most $K$ y-coordinates that minimizes
\begin{equation}
   \lambda\cdot|S| + \sum_{i=1}^n \min_{y \in S} |y - p_i|.
\label{EQ:k-median}
\end{equation}
Note that we can optimize the value in Eq.~(\ref{EQ:k-median}) by
choosing $S \subseteq \{p_1,\dots,p_n\}$:
For $y \in S \setminus
\{p_1,\dots,p_n\}$ let $\{p_i,\dots,p_j\}$ be the set of points that
we would connect to the backbone through $y$. Let $p_i > \dots >
p_{i'} > y > p_{i'+1} >\dots > p_j$. If $i' - i + 1 \geq j - i'$, replace
$y$ be $p_{i'}$. Otherwise replace $y$ be $ p_{i'+1}$.  Then the
objective value in Eq.~\ref{EQ:k-median} can at most improve.
Hence, the problem can be solved in $O(Kn)$ time if the points are
sorted according to their y-coordinates using the algorithm of
Hassin and Tamir~\cite{Hassin1991395}. Note that the problem
corresponds to the $K$-median problem if $\lambda = 0$.

\paragraph{Multiple Colors.}
If the $n$ points have different colors, we can no longer assume
that all backbones go through one of the given $n$ points. However,
by Lemma~\ref{lemma:intermediate-colors}, it suffices to add between
any pair of vertically consecutive points two additional candidates
for backbone positions, plus one additional candidate above all
points and one below all points.  Hence, we have a set of $3n$
candidate lines at y-coordinates
\begin{equation}
   p_1^- > p_1 > p_1^+ > p_2^- > p_2 > p_2^+ > \dots > p_n^- > p_n > p_n^+
   \label{EQ:candidate_lines}
\end{equation}
\begin{figure}
  \centerline{\includegraphics{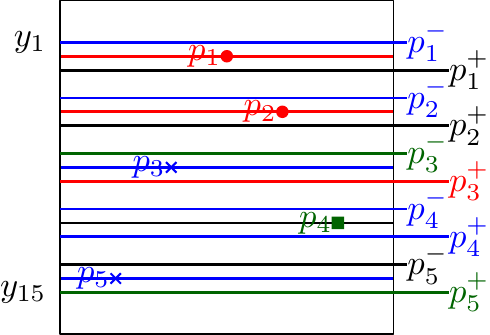}}
\caption{\label{FIG:candidates}Candidates for five points. Red points are circles, blue points are crosses, and the green point is a square. Candidates through a point have the same color as the point. Candidate $p_i^\pm$ has the same color as the first point with a different color as $p_i$ that is met when walking from $p_i^\pm$ over $p_i$. Candidates $p_1^+$, $p_2^+$, and $p_5^-$ will not be used and have no color.}
\end{figure}
where for each $i$ the values $p_i^-$ and $p_i^+$ are as close to
$p_i$ as the label heights allow. Clearly, a backbone through $p_i$
can only be connected to points with color $c(p_i)$.
If we use a backbone through $p_i^-$ (or $p_i^+$, respectively), it
will have the same color as the first point below $p_i$ (or above
$p_i$, respectively) that has a different color than $p_i$.
E.g., in Fig~\ref{FIG:candidates}, $p_1^-$ is colored blue, since $p_3$ is the first point below
$p_1$ that has a different color than red, namely blue.
Hence, the colors of all candidates are fixed or the candidate will
never be used as a backbone.
For an easier notation, we denote the $i$th point in
Eq.~\ref{EQ:candidate_lines} by $y_i$ and its color by $c(y_i)$.

We solve the problem using dynamic programming. For each
$i=1,\dots,3n$, and for each vector $\vec k' =
(k'_1,\dots,k'_{|C|})$ with $k'_1 \leq k_1,\dots,k'_{|C|} \leq
k_{|C|}$, let $L(i,\vec k')$ denote the minimum length of a feasible
backbone labeling of $p_1,\dots,p_{\lfloor\frac{i+1}{3}\rfloor}$
using $k'_c$ infinite backbones of color $c$ for $c=1,\dots,|C|$
such that the bottommost backbone is at position $y_i$, if such a
labeling exists. Otherwise $L(i,\vec k') = \infty$. In the
following, we describe, how the values $L(i,\vec k')$ can be
computed recursively in $\mathcal O(n^2\prod_{i=1}^{|C|} k_i)$ time
in total.

\begin{lemma}
All values $L(i,\vec k')$ can be computed in $\mathcal
O(n^2\prod_{i=1}^{|C|} k_i)$ time in total.
\end{lemma}
\begin{proof}
Assume that
we want to place a new backbone at $y_i$ and that the previous
backbone was at $y_j,j < i$. Then we have to connect the points
$p_x, {(j+2)}/{3} \leq x \leq {i}/{3}$ between $y_i$ and $y_j$ to
one of the backbones through $y_i$ or $y_j$. Let
$\textsc{link}(j,i)$ denote the minimum total length of the vertical
segments linking these points to their respective backbone. If there
is a point $p_x$ between $y_i$ and $y_j$ with $c(p_x) \notin
\{c(y_i),c(y_j)\}$, we set $\textsc{link}(j,i)=\infty$. Otherwise,
we have
\begin{equation}
   \textsc{link}(j,i) =
     \begin{cases}
       \displaystyle\sum_{\frac{j+2}{3} \leq x \leq \frac{i}{3}}\min\{y_j - y_x,y_x - y_i\} &
                           \textup{if } c(y_i) = c(y_j)\\
        \displaystyle\sum_{\frac{j+2}{3} \leq x \leq \frac{i}{3} \atop c(p_x) = c(y_j)}(y_j - y_x) +
        \displaystyle\sum_{\frac{j+2}{3} \leq x \leq \frac{i}{3} \atop c(p_x) = c(y_i)}(y_x - y_i) &
                           \textup{if } c(y_i) \neq c(y_j)\\
     \end{cases}
\label{EQ:link}
\end{equation}

With the base cases $L(i,0,\dots,0,k'_i = 1,0,\dots,0) = \sum_{0 < x
\leq i/3} (y_i - p_x)$ if all points above $y_i$ have the same color
as $y_i$ and
$\infty$ otherwise, as well as $L(i,\vec{k'}) = \infty$ if $k'_i =
0$, we get the following recursion

\begin{equation}
    L(i,k'_1,\dots,k'_{|C|}) = \lambda
    + \min_{j \leq i}
    \left(
        L(j,k'_1,\dots,k'_{c(y_i)} - 1,\dots,k'_{|C|})
       + \textsc{link}(j,i)
    \right).
\label{EQ:multi-color-length}
\end{equation}

It remains to show that for a fixed index $i$, all values
\textsc{link}$(j,i)$, $j<i$ can be computed in $O(n)$ time. Let $c'$
be the first color of a point above $y_i$ that is different from
$c(y_i)$. For a fixed $i$, starting from $j=i-1$, we scan the
candidates in decreasing order of their indices until we find the
first point that is neither colored $c'$ nor $c(y_i)$.  For each
$j$, let $P'_i(j) = \{p_x;\;\frac{j+2}{3} \leq x \leq \frac{i}{3},
c(p_x) = c'\}$, $P_i(j) = \{p_x;\;\frac{j+2}{3} \leq x \leq
\frac{i}{3}, c(p_x) = c(y_i)\}$, $P_i^\text{up}(j) =
\{p_x;\;\frac{j+2}{3} \leq x \leq \frac{i}{3}, c(p_x) = c(y_i),
y_j-p_x \leq p_x-y_i\}$.  In each step, we can update in constant
time the values %
$|P'_i(j)|$, %
$\textsc{secondLength}(j) = \sum_{p\in P'_i(j)} (y_j - p) =$
$\textsc{secondLength}(j-1) + |P'_i(j)|\cdot (y_j - y_{j-1})$, and
$\textsc{firstLength}(j) = \sum_{p\in P_i(j)} (p - y_i)$. %
Further, we can update in amortized constant time the index of the
lowest element in $P_i^\text{up}(j)$ as well as the values
$|P_i^\text{up}(j)|$, %
$\textsc{firstDownLength}(j) = \sum_{p\in P_i(j)\setminus
  P_i^\text{up}(j) } (p - y_i)$, and also
$\textsc{firstUpLength}(j)= \sum_{p\in P_i^\text{up}(j) } (y_j -
p)$. %
With these values we can directly compute \textsc{link}$(j,i)$
according to the case distinction in Eq.~\ref{EQ:link}. Hence, we
can compute a value $L(i,k'_1,\dots,k'_{|C|})$ according to the
recursion in Eq.~\ref{EQ:multi-color-length} in $O(n)$ time.
Therefore, all values $L(i,\vec k')$ can be computed in $\mathcal
O(n^2\prod_{i=1}^{|C|} k_i)$ time in total.\qed
\end{proof}

Let $S$ be the set of candidates $y_i$ such that all points below
$y_i$ have the same color as $y_i$. Then we can compute the minimum
total length of a backbone  labeling of $p_1,\dots,p_n$ with at most
$k_c,c=1,\dots,|C|$ labels per color $c$ as
\[
  \min_{y_i \in S\cup\{p^{+}_n\},k'_1\leq k_1,\dots,k'_{|C|}\leq
k_{|C|}}\left(L(i,k'_1,\dots,k'_{|C|}) + \sum_{\frac{i+2}{3} \leq x \leq n} (y_i - p_x) \right).
\]

If we globally bound the total number of labels by $K$, we obtain a
similar dynamic program with the corresponding values $L(i,k)$,
$i=1,\dots,3n$, $k<K$. We summarize the results in the following
theorem.
\begin{theorem}
A minimum length backbone labeling with infinite backbones for $n$
points with $|C|$ colors can be computed in $O(n^2 \cdot
\prod_{i=1}^{|C|} k_i)$ time if at most $k_i$ labels are allowed for
color $i$, $i=1,\dots,|C|$ and in $O(n^2 \cdot K)$ time if in total
at most $K$ labels are allowed.
\end{theorem}

\subsection{Finite Backbones.}
We modify the dynamic program used for minimizing the total number
of labels (in Section~\ref{sec:finite-backbones-label-min}) for
minimizing the total leader length. First, we now denote by the
$T$-values the additional length of segments and backbones needed
for labeling the points of the subinstance. We do, however, have to
apply more changes. By the case of a single point connected to a
backbone, we see that we have to allow backbones passing through
input points of the same color. Additionally, for computing the
vertical length needed for connecting to a backbone placed in a gap,
we need to know its actual $y$-coordinate.

Suppose there is a set $B$ of backbones that all lie in the same gap
between points $p_i$ and $p_{i+1}$.  Let $b^\star$ be the longest of
those; see Fig.~\ref{fig:cost-finite-backbones-split}. $b^\star$
vertically splits the set $B$; any backbone $b' \in B$ above
$b^\star$ connects only to points above it and any backbone $b'' \in
B$ below $b^\star$ connects only to points below it. By moving $b'$
to the top and $b''$ to the bottom as far as possible the total
leader length decreases. Hence, in any optimum solution, the
backbones above $b^\star$ will be very close to $y(p_i)$ and the
backbones below $b^\star$ will be very close to $y(p_{i+1})$.
Furthermore, depending on the numbers of connected points above and
below, by either moving $b^\star$ to the top or to the bottom we
will find a solution that is not worse, and in which any backbone of
$B$ is close to $p_i$ or $p_{i+1}$.

If we allow backbones to be infinitely close to points or other
backbones, we can %
use backbone positions $p_{i}^-$ and $p_{i}^+$ that lie infinitely
close above and below $p_{i}$, respectively, and share its
y-coordinate. Each of these positions may be used for an arbitrary
number of backbones. With these positions as well as the input
points as possible label positions, we can then find a solution with
minimum total leader length in $O(n^4 |C|^2)$ time if the number of
labels is not bounded by adding the length of the newly placed
segments in any calculation.

\paragraph{Bounded numbers of labels.} If we want to integrate an
upper bound $K$ on the total number of labels, or, for each color $c
\in C$, an upper bound $k_c$ on the number of labels of color $c$,
into the dynamic program, we need an additional dimension for the
remaining number of backbones that we can use in the subinstance (or
a dimension for each color $c \in C$ for the remaining number of
backbones of that color). Additionally, when splitting the instance
into two parts, we have to consider not only the position of the
splitting backbone of color $c$, but also the different combinations
of distributing the allowed number(s) of backbones among the
subinstances. For a global bound $K$, we need $O(nK)$ time for
computing an entry of the table. If we have individual bounds $k_c$
for $c\in C$, we need $O(n \prod_{c \in C} k_c)$ time. Together with
the additional dimension(s) of the table, we can minimize the total
leader length in $O(n^4|C|^2K^2)$ time if we have a global bound
$K$, and in $O\left(n^4 |C|^2 \left( \prod_{c \in C} k_c \right)^2
\right)$ time if we have an individual bound $k_c$ for each color $c
\in C$.

\paragraph{Minimum distances.}
So far, we allowed backbones to be infinitely close to unconnected
points and other backbones, which will, in practice, lead to
overlaps. One would rather enforce a small distance between two
backbones or a backbone and a point, even if this increases the
total leader length a bit. Let $\Delta>0$ be the minimum allowed
distance. In an optimum solution, there will be two sequences of
backbones on the top and on the bottom of a gap between $p_i$ and
$p_{i+1}$, such that inside a sequence consecutive backbones have
distance $\Delta$; see
Fig.~\ref{fig:cost-finite-backbones-min-dist}. We get all possible
backbone positions inside the gap by taking all y-coordinates inside
whose y-distance to either $p_i$ or $p_{i+1}$ is an integer multiple
of $\Delta$; see Fig.~\ref{fig:cost-finite-backbones-candidates}.
Note that $n$ positions of each type suffice in a gap; if the gap is
too small, there might even be less positions. The two sequences can
overlap. In this case, we have to check that we do not combine two
positions with a distance smaller than $\Delta$ in the dynamic
program.

Together with the input points, we get a set of $O(n^2)$ candidate
positions for backbones, each of which can be used at most once.
This increases the number of entries of table $T$ by a factor of
$O(n^2)$, and the running time of computing a single entry by a
factor of $O(n)$. The resulting running time of our dynamic program
is $O(n^7|C|^2)$ if we do not bound the number of labels,
$O(n^7|C|^2K^2)$ if we have a global bound $K$ on the number of
labels, and $O(n^7 |C|^2 ( \prod_{c \in C} k_c )^2)$ if we have an
individual bound $k_c$ for each color $c \in C$.

\begin{figure}[tb]
  \centering
  \begin{subfigure}[t]{.32\textwidth}
    \centering
    \includegraphics[page=2]{images/cost-finite-backbones}
    \caption{}
    \label{fig:cost-finite-backbones-split}
  \end{subfigure}
  \hfill
  \begin{subfigure}[t]{.32\textwidth}
    \includegraphics[page=3]{images/cost-finite-backbones}
    \caption{}
    \label{fig:cost-finite-backbones-min-dist}
  \end{subfigure}
  \hfill
  \begin{subfigure}[t]{.32\textwidth}
    \includegraphics[page=1]{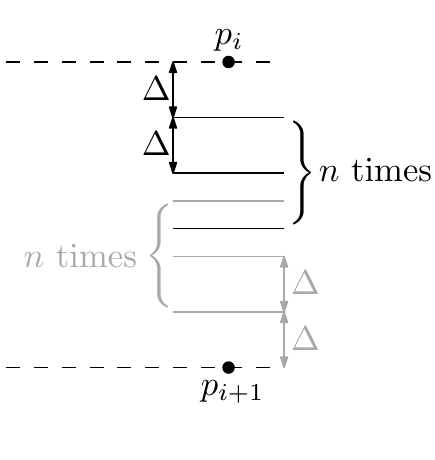}
    \caption{}
    \label{fig:cost-finite-backbones-candidates}
  \end{subfigure}
  \caption{(a)~Longest backbone $b^\star$ splitting the backbones between $p_i$ and $p_{i+1}$. (b)~Backbones placed with the minimum leader length. (c)~Candidate positions for backbones inside the gap.}
\end{figure}

\begin{theorem}
Given a set $P$ of $n$ colored points, a color set $C$, and a label
bound $K$ (or color vector $\vec k$), we can compute a feasible
labeling of $P$ with finite backbones that minimizes the total
leader length in time $O(n^7|C|^2K^2)$ (or $O(n^7 |C|^2 ( \prod_{c
\in C} k_c )^2)$). \label{thm:lengthfinite}
\end{theorem}

\section{Crossing Minimization}\label{sec:cross}
In this section we allow crossings between backbone leaders, which
generally allows us to use fewer labels. We concentrate on
minimizing the number of crossings for the case $K=|C|$, i.e., one
label per color, and distinguish fixed and flexible label orders.

\subsection{Fixed y-Order of Labels}

In this part we assume that the color set $C$ is ordered and we
require that for each pair of colors $i<j$  the label (and backbone)
in color $i$ is above the label in color $j$.

\subsubsection{Infinite Backbones.}
\label{sec:crossing-min-fixed}
\label{sec:infinite-crossing-min}
Observe that it is always possible to slightly shift the backbones
of a solution without increasing the number of crossings such that
no backbone contains a point. Thus, the backbones can be assumed to
be positioned in the gaps between vertically adjacent points; we
number the gaps from $0$ to $n$ as in Section~\ref{sec:number}.

Suppose that we fix the position of the $i$-th backbone to gap~$g$.
For $1 \leq i \leq |C|$ and $0 \leq g \leq n$, let $\cross(i,g)$ be
the number of crossings of the vertical segments of the
non-$i$-colored points when the color-$i$ backbone is placed at gap
$g$. Note that this number depends only on the y-ordering of the
backbones, which is fixed, and not on their actual positions. So, we
can pre-compute table $\cross$, using dynamic programming, as
follows. All table entries of the form $\cross(\cdot,0)$ can be
clearly computed in $O(n)$ time. Then,
$\cross(i,g)=\cross(i,g-1)+1$, if the point between gaps $g-1$ and
$g$ has color $j$ and $j>i$. In the case where the point between
gaps $g-1$ and $g$ has color $j$ and $j<i$,
$\cross(i,g)=\cross(i,g-1)-1$. If it has color $i$, then
$\cross(i,g)=\cross(i,g-1)$. From the above, it follows that the
computation of table $\cross$ takes $O(n|C|)$ time.

Now, we use another dynamic program to compute the minimum number of
crossings.  Let~$T[i,g]$ denote the minimum number of crossings on
the backbones~$1,\dots, i$ in any solution subject to the condition
that the backbones are placed in the given ordering and backbone~$i$
is positioned in gap~$g$.  Clearly~$T[0,g] = 0$ for~$g=0,\dots, n$.
Moreover, we have~$T[i,g] = \min_{g' \le g} T[i-1,g'] +
\cross(i,g)$. Having pre-computed table
$\cross$ and assuming that for each entry $T[i,g]$, we also
store the smallest entry of row $T[i,\cdot]$ to the left of $g$, each entry
of table~$T$ can by computed in constant time. Hence, table~$T$ can
be filled in time~$O(n|C|)$. Then, the minimum crossing number
is~$\min_g T[|C|,g]$.  A corresponding solution can be found by
backtracking in the dynamic program.

\begin{theorem}
Given a set $P$ of $n$ colored points and  an ordered color set $C$,
a backbone labeling with one label per color, labels in the given
color order,  infinite backbones, and minimum number of crossings
can be computed in $O(n |C|)$ time.
\label{theorem:min-crossings-fixed-order}
\end{theorem}

\subsubsection{Finite Backbones.}
We can easily modify the approach used for infinite backbones to
minimize the number of crossings for finite backbones, if the
y-order of labels is fixed, as the following theorem shows.

\begin{theorem}
Given a set $P$ of $n$ colored points and an ordered color set $C$,
a backbone labeling with one label per color, labels in the given
order, finite backbones, and minimum number of crossings can be
computed in $O(n |C|)$ time.
\label{theorem:min-crossings-fixed-order-finite-backbones}
\end{theorem}
\begin{proof}
We present a dynamic program similar to the one presented in the
proof of Theorem~\ref{theorem:min-crossings-fixed-order}. Recall
that all points of the same color are routed to the same label and
the order of the labels is fixed, i.e., the label of the $i$-th
colored points is above the label of the $j$-th colored points, when
$i<j$. Hence, the computation of the number of crossings when fixing
a backbone at a certain position should take into consideration that
the backbones are not of infinite length. Recall that the dynamic
program could precompute these crossings, by maintaining an $n
\times |C|$ table $\cross$, in which each entry $\cross(i,g)$
corresponds to the number of crossings of the non-$i$-colored points
when the color-$i$-backbone is placed at gap $g$, for $1 \leq i \leq
|C|$ and $0 \leq g \leq n$. In our case,
$\cross(i,g)=\cross(i,g-1)+1$, if the point between gaps $g-1$ and
$g$ right of the leftmost $i$-colored point and has color $j$ s.t.
$j>i$.  In the case, where the point between gaps $g-1$ and $g$ is right of the leftmost $i$-colored point and has color $j$ and $j<i$,
$\cross(i,g)=\cross(i,g-1)-1$. Otherwise,
$\cross(i,g)=\cross(i,g-1)$. Again, all table entries of the form
$\cross(\cdot,0)$ can be clearly computed in $O(n)$ time.\qed
\end{proof}

\subsection{Flexible y-Order of Labels}
In this part the order of labels is no longer given and we need to
minimize the number of crossings over all label orders. While there
is an efficient algorithm for infinite backbones, the problem is
NP-complete for finite backbones.

\subsubsection{Infinite Backbones.}
We give an efficient algorithm for the case that there are $K=|C|$
fixed label positions $y_1, \dots, y_K$ on the right boundary of
$R$, e.g., uniformly distributed.

\begin{theorem}
Given a set $P$ of $n$ colored points, a color set $C$, and a set of
$|C|$ fixed label positions, we can compute in $O(n + {|C|}^3)$ time
a feasible backbone labeling with infinite backbones that minimizes
the number of crossings.
\end{theorem}
\begin{proof}
First observe that if the backbone of color $k, 1 \leq k \leq |C|$
is placed at position $y_i, 1 \leq i \leq |C|$, then the number of
crossings created by the vertical segments leading to this backbone
is fixed, since all label positions will be occupied by an infinite
backbone. This crossing number $\mathrm{cr}(k,i)$ can be determined
in $O(n_k + |C|)$ time, where $n_k$ is the number of points of color
$k$.
In fact, by a sweep from top to bottom, we can even determine all
crossing numbers $\mathrm{cr}(k,\cdot)$ for backbone $k, 1 \leq k
\leq {|C|}$ in time $O(n_k + {|C|})$. Now, we construct an instance
of a weighted bipartite matching problem, where for each position
$y_i, 1 \leq k \leq {|C|}$ and each backbone $k, 1 \leq k \leq
{|C|}$, we establish an edge $\{k,i\}$ of weight $cr(k,i)$. In
total, this takes $O(n + {|C|}^2)$ time. The minimum-cost weighted
bipartite matching problem can be computed in time $O({|C|}^3)$ with
the Hungarian method~\cite{k-hmap-55} and yields a backbone labeling
with the minimal number of crossings.\qed \end{proof}

\subsubsection{Finite Backbones.}
Next, we consider the variant with finite backbones and prove that
it is NP-hard to minimize the number of crossings. For simplicity,
we allow points that share the same x- or y-coordinates.  This can
be remedied by a slight perturbation.  Our arguments do not make use
of this special situation, and hence carry over to the perturbed
constructions. We first introduce a number of gadgets that are
required for our proof and sketch their properties, before
describing the  hardness reduction.

The construction consists of the middle backbone, whose position
will be restricted to a given range~$R$, and an upper and a lower
\emph{guard gadget} that ensure that positioning the middle backbone
outside range~$R$ creates many crossings.  We assume that the middle
backbone is connected to at least one point further to the left such
that it extends beyond all points of the guard gadgets.  The middle
backbone is connected to two \emph{range points} whose y-coordinates
are the upper and lower boundary of the range~$R$.
Their~x-coordinates are such that they are on the right of the
points of the guard gadgets.  A \emph{guard} consists of a backbone
that connects to a set of~$M$ points, where~$M > 1$ is an arbitrary
number.  The~$M$ points of a guard lie left of the range points. The
upper guard points are horizontally aligned and lie slightly below
the upper bound of range~$R$.  The lower guard points are placed
such that they are slightly above the lower bound of range~$R$.  We
place~$M$ upper and $M$ lower guards such that the guards form pairs
for which the guard points overlap horizontally. The upper (resp.\
lower) guard gadget is formed by the set of upper (resp.\ lower)
guards.  We call~$M$ the \emph{size} of the guard gadgets. The next
lemma shows the properties of the range restrictor.

\begin{figure}[tb]
  \centering
  \begin{subfigure}[b]{.45\textwidth}
    \centering
    \includegraphics[page=1]{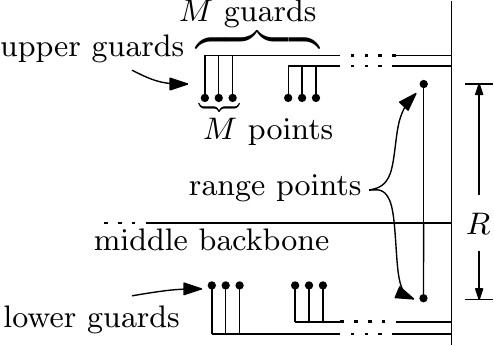}
    \caption{}
    \label{fig:restrictor}
  \end{subfigure}\hfil
  \begin{subfigure}[b]{.3\textwidth}
    \includegraphics[page=3]{images/restrictor-gadget}
    \caption{}\label{fig:blocker}
  \end{subfigure}
  \caption{The range restrictor gadget (a), and a blocker gadget (b).}
\end{figure}

\begin{lemma}
The backbones of the range restrictor can be positioned such that
there are no crossings.  If the middle backbone is positioned
outside the range~$R$, there are at least~$M-1$ crossings.
\label{lem:range-restrictor}
\end{lemma}
\begin{proof}
The first statement is illustrated by the drawing in
Fig.~\ref{fig:restrictor}.  It remains to show that positioning the
middle backbone outside range~$R$ results in at least~$M-1$
crossings.  Suppose for a contradiction that the middle backbone is
positioned outside range~$R$ and that there are fewer than~$M-1$
crossings.  Assume without loss of generality that the middle
backbone is embedded below range~$R$; the other case is symmetric.

First, observe that all guards must be positioned above the middle
backbone, as a guard below the middle backbone would create~$M$
crossings, namely between the middle backbone and the segments
connecting the points of the guard to its backbone.  Hence the
middle backbone is the lowest.  Now observe that any guard that is
positioned below the upper range point crosses the segment that
connects this range point to the middle backbone.  To avoid
having~$M-1$ crossings, it follows that at least~$M+1$ guards (both
upper and lower) must be positioned above range~$R$. Hence, there is
at least one pair consisting of an upper and a lower guard that are
both positioned above the range~$R$.  This, however, independent of
their ordering, creates at least~$M-1$ crossings; see
Fig.~\ref{fig:guard-pair}, where the two alternatives for the lower
guard are drawn in black and bold gray, respectively.  This
contradicts our assumption.\qed
\end{proof}

\begin{figure}[h]
  \centering
  \includegraphics[page=2]{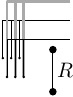}
  \caption{Crossings caused by a pair of an upper and a lower guard that are positioned on the same side outside range $R$.}
  \label{fig:guard-pair}
\end{figure}

Let~$B$ be an axis-aligned rectangular box and let~$R$ be a small
interval that is contained in the range of~y-coordinates spanned
by~$B$.  A \emph{blocker gadget} of \emph{width} $m$ consists of a
backbone that connects to~$2m$ points, half of which are positioned on
the top and bottom side of~$B$, respectively.  Moreover, a range
restrictor gadget is used to restrict the backbone of the blocker to
the range~$R$.  Figure~\ref{fig:blocker} shows an example.  Note that,
due to the range restrictor, this drawing is essentially fixed.  We
say that a backbone \emph{crosses} the blocker gadget if its backbone
crosses the box~$B$.  It is easy to see that any backbone that crosses
a blocker gadget creates~$m$ crossings, where~$m$ is the width of the
blocker.

We are now ready to show that the crossing minimization problem with
flexible y-order of the labels is NP-hard.

\begin{theorem}\label{thm:np-hard}
Given a set of $n$ input points in $k$ different colors and an
integer $Y$ it is NP-complete to decide whether a backbone labeling
with one label per color and flexible y-order of the labels that has
at most $Y$ leader crossings exists.
\end{theorem}
\begin{proof}
The proof is by reduction from the NP-complete Fixed Linear
Crossing Number problem~\cite{mnkf-cmleg-90}, which is defined as
follows. Given a graph $G=(V,E)$, a bijective function $f \colon V
\rightarrow \{1, \dots, |V|\}$, and an integer $Z$, is there a
drawing of~$G$ with the vertices placed on a horizontal line (the
\emph{spine}) in the order specified by~$f$ and the edges drawn as
semi-circles above or below the spine so that there are at most $Z$
edge crossings? Masuda et al.~\cite{mnkf-cmleg-90} showed that the
problem remains NP-complete even if $G$ is a matching.

\begin{figure}[htb]
  \centering
  \includegraphics[scale=1,page=3]{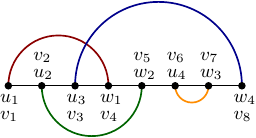}
  \caption{An input instance with four edges.}
  \label{fig:4-edge-input}
\end{figure}

Let $G$ be a matching. Then the number of vertices is even and we
can assume that the vertices $V=\{v_1, \dots, v_{2n}\}$ are indexed
in the order specified by $f$, i.e., $f(v_i) = i$ for all $i$.
Furthermore, we direct every edge $\{v_i,v_j\}$ with $i<j$ from
$v_i$ to $v_j$.
Let $\{u_1, \dots, u_n\}$ be the ordered source vertices and let
$\{w_1, \dots, w_n\}$ be the ordered sink vertices.
Figure~\ref{fig:4-edge-input} shows an example graph~$G$ drawn on a
spine in the specified order.

In our reduction we will create an edge gadget for every edge in
$G$. The gadget consists of five blocker gadgets and one \emph{side
selector gadget}. Each of the six sub-gadgets uses its own color and
thus defines one backbone. The edge gadgets are ordered from left to
right according to the sequence of source vertices $(u_1, \dots,
u_n)$. Figure~\ref{fig:4-edge-gadget} shows a sketch of the instance
$I_G$ created for a matching~$G$ with four edges.

\begin{figure}[tb]
    \centering
    \includegraphics[scale=1,page=4]{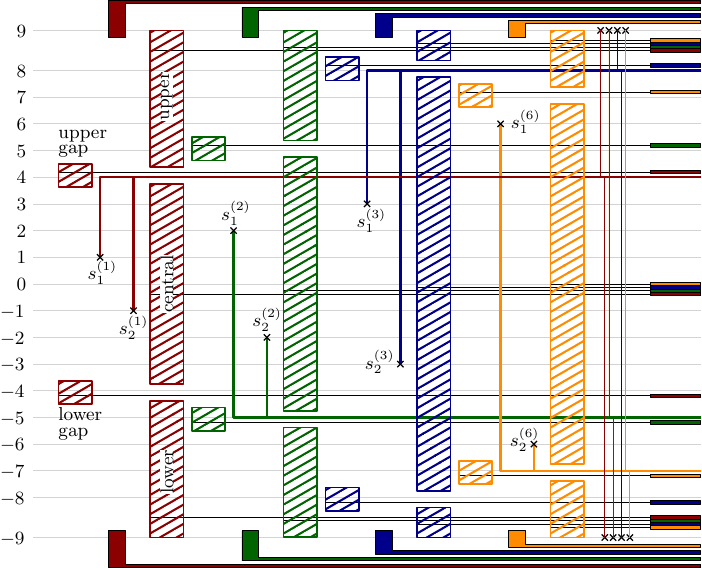}
    \caption{Sketch of the reduction of the graph of Fig.~\ref{fig:4-edge-input} into a backbone labeling instance. Hatched rectangles represent blockers, thick segments represent side selectors, and filled shapes represent guard gadgets or range restrictor gadgets.}
    \label{fig:4-edge-gadget}
\end{figure}

The edge gadgets are placed symmetrically with respect to the
x-axis. We create $2n+1$ special rows above the x-axis and $2n+1$
special rows below, indexed by $-(2n+1),-2n, \dots,0,\dots,
2n,2n+1$. The gadget for an edge $(v_i,v_j)$ uses five blocker
gadgets (denoted as \emph{central, upper, lower, upper gap}, and
\emph{lower gap} blockers) in two different columns to create two
small gaps in rows $j$ and $-j$, see the hatched blocks in the same
color in Fig.~\ref{fig:4-edge-gadget}. The upper and lower blockers
extend vertically to rows $2n+1$ and $-2n-1$. The gaps are intended
to create two alternatives for routing the backbone of the side
selector. Every backbone that starts left of the two gap blockers is
forced to cross at least one of these five blocker gadgets as long
as it is vertically placed between rows $2n+1$ and $-2n-1$. The
blockers have width $m=8 n^2$. Their backbones are fixed to lie
between rows $0$ and $-1$ for the central blocker, between rows $2n$
and $2n+1$ ($-2n$ and $-2n-1$) for the upper (resp.\ lower) blocker,
and between rows $j$ and $j+1$ ($-j$ and $-j-1$) for the upper
(resp.\ lower) gap blocker.

The \emph{side selector} consists of two horizontally spaced
\emph{selector points} $s_1^{(i)}$ and $s_2^{(i)}$ in rows $i$ and
$-i$ located between the left and right blocker columns. They have
the same color and thus define one joint backbone that is supposed
to pass through one of the two gaps in an optimal solution. The~$n$
edge gadgets are placed from left to right in the order of their
source vertices, see Fig.~\ref{fig:4-edge-gadget} for an example.

The backbone of every selector gadget is vertically restricted to
the range between rows~$2n+1$ and~$-2n-1$ in any optimal solution by
augmenting each selector gadget with a range restrictor gadget. This
means that we add two more points for each selector to the right of
all edge gadgets, one in row $2n+1$ and the other in row $-2n-1$.
They are connected to the selector backbone. In combination with a
corresponding upper and lower guard gadget of size $M = \Omega(n^4)$
between the two selector points $s_1^{(i)}$ and $s_2^{(i)}$ this
achieves the range restriction according to
Lemma~\ref{lem:range-restrictor}.

\begin{numclaim}
In a crossing-minimal labeling the backbone of the selector gadget
for every edge $(v_i,v_j)$ passes through one of its two gaps in
rows $j$ or $-j$. \label{clm:gaps}
\end{numclaim}
\begin{proof}
There are basically three different options for placing a selector
backbone: (a)~outside its range restriction, i.e., above row $2n+1$
or below row $-2n-1$, (b) between rows $2n+1$ and $-2n-1$, but
outside one of the two gaps, and (c) in rows $j$ or $-j$, i.e.,
inside one of the gaps. In case (a) we get at least $M= \Omega(n^4)$
crossings by Lemma~\ref{lem:range-restrictor}. So we may assume that
case~(a) never occurs for any selector gadget; we will see that in
this case there are only $O(n^4)$ crossing in total for the selector
gadgets. In cases (b) and~(c) we note that the backbone will cross
one blocker for each edge whose source vertex is right of $v_i$ in
the order $(u_1, \dots, u_n)$. Let $k$ be the number of these edges.
Additionally, in case (b), the backbone crosses one of its own
blockers. In cases (b) and (c) the two vertical segments of the
range restrictor of edge $(v_i,v_j)$ cross every selector and
blocker backbone regardless of the position of its own backbone,
which yields $6n-1$ crossings. Thus, case (b) causes at least
$(k+1)\cdot m + 6n - 1$ crossings.

To give an upper bound on the number of crossings in case (c) we
note that the backbone can cross at most three vertical segments of
any other selector gadget, the two segments connected to its
selector points and one segment connected to a point in either row
$2n+1$ or $-2n-1$, which is part of the range restrictor gadget. The
two vertical segments connected to points $s_1^{(i)}$ and
$s_2^{(i)}$ together will cross the backbone of each central blocker
at most once, the backbones of each pair of upper/lower gap blockers
at most twice, and each selector backbone at most twice. Backbones
of upper and lower blockers are never crossed in case (c). So in
case (c) the segments of the selector gadget cross at most $km +
8n-1$ segments, which is less than the lower bound of $(k+1)m +
6n-1$ in case (b). We conclude that each backbone indeed passes
through one of the gaps in an optimal solution. Any violation of
this rule would create at least $m$ additional crossings, which is
more than what an arbitrary assignment of selector backbones to gaps
yields. \qed
\end{proof}

Next, we show how the number of crossings in the backbone labeling
instance relates to the number of crossings in the Fixed Linear
Crossing Number problem. There is a large number of unavoidable
crossings regardless of the backbone positions of the selector
gadgets. By Claim~\ref{clm:gaps} and the fact that violating any
range restriction immediately causes $M$ crossings, we can assume
that every backbone adheres to the rules, i.e., stays within its
range as defined by the range restriction gadgets or passes through
one of its two gaps.

\begin{numclaim}
An optimal solution of the backbone labeling instance $I_G$ created
for a matching $G$ with $n$ edges has $X+2Z$ crossings, where $X$ is
a constant depending on $G$ and $Z$ is the minimum number of
crossings of $G$ in the Fixed Linear Crossing Number
problem.\label{clm:crossings}
\end{numclaim}
\begin{proof}
Aside from guard backbones, which never have crossings, there are
two types of backbones in our construction, blocker and selector
backbones. We argue separately for all four possible types of
crossings and distinguish fixed crossings that must occur and
variable crossings that depend on the placement of the selector
backbones.

\begin{enumerate}[(I)]
    \item crossings between blocker backbones and vertical blocker segments,
    \item crossings between blocker backbones and vertical selector segments,
    \item crossings between selector backbones and vertical blocker segments, and
    \item crossings between selector backbones and vertical selector segments.
\end{enumerate}

\textbf{Case I:} By construction each blocker backbone must
intersect exactly one blocker gadget of width $m$ for each edge
gadget to its right. Thus we obtain $X_1 = 5m \sum_{i=1}^{n-1} i =
5m (n^2-n)/2$ fixed crossings.

\textbf{Case II:} Each blocker backbone crosses for each edge
exactly one vertical selector segment that is part of the range
restrictor gadget on the right-hand side of our construction. Each
central blocker backbone additionally crosses for each edge gadget
to its right one vertical segment incident to one of the selector
points, regardless of the selector position. The two gap blocker
backbones for gaps in rows $j$ and $-j$ together cause two
additional crossings for each edge gadget to its right whose target
vertex $v_k$ satisfies $k>j$. To see this we need to distinguish two
cases. Let $e=(v_i,v_k)$ be the edge of an edge gadget with $k>j$.
If $i<j$, then both vertical selector segments either cross the
lower gap blocker backbone or they both cross the upper gap blocker
backbone (see edges $(v_1,v_4)$ and $(v_2,v_5)$ in
Fig.~\ref{fig:4-edge-gadget}). If $i>j$, then one of the two
vertical selector segments crosses both gap blocker backbones, and
the other one crosses none (see edges $(v_1,v_4)$ and $(v_6,v_7)$ in
Fig.~\ref{fig:4-edge-gadget}). The backbones of the upper and lower
blockers do not cross any other vertical selector segments.

Let $\kappa = |\{\{(v_i,v_j),(v_k,v_l)\} \in E^2 \mid i < k \text{
and } j<l\}| = O(n^2)$. Then we obtain $X_2 = 5n^2 + (n^2 -n)/2 +
2\kappa$ fixed crossings from Case II.

\textbf{Case III:} Each selector backbone that passes through one of
its gaps crosses exactly one blocker gadget for each edge gadget to
its right. Thus we obtain $X_3 = m(n^2-n)/2$ fixed crossings in Case
III.

\textbf{Case IV:} Let $e=(v_i,v_j)$ and $f=(v_k,v_l)$ be two edges
in $G$, and let $i<k$. Then there are three sub-cases: (a) $e$ and
$f$ are \emph{sequential}, i.e., $i<j<k<l$, (b) $e$ and $f$ are
\emph{nested}, i.e., $i<k<l<j$, or (c) $e$ and $f$ are
\emph{interlaced}, i.e., $i<k<j<l$. For every pair of sequential
edges there is exactly one crossing, regardless of the gap
assignments (see edges $(v_1,v_4)$ and $(v_6,v_7)$ in
Fig.~\ref{fig:4-edge-gadget}). For every pair of nested edges there
is no crossing, regardless of the gap assignments (see edges
$(v_3,v_8)$ and $(v_6,v_7)$ in Fig.~\ref{fig:4-edge-gadget}).
Finally, for every pair of interlaced edges there are no crossings
if they are assigned to opposite sides of the x-axis or two
crossings if they are assigned to the same side. So interlaced edges
do not contribute to the number of fixed crossings. Let $\tau =
|\{\{(v_i,v_j),(v_k,v_l)\} \in E^2 \mid i < j < k<l\}| = O(n^2)$ be
the number of sequential edge pairs. Then we obtain $X_4 = \tau$
fixed crossings from Case IV.

From the discussion of the four cases we can  immediately see that
all  crossings are fixed, except for those related to interlaced
edge pairs (see for example edges $(v_1,v_4)$ and $(v_3,v_8)$ or
$(v_2,v_5)$ in Fig.~\ref{fig:4-edge-gadget}). But these are exactly
the edge pairs that  create crossings in the Fixed Linear Crossing
Number problem if assigned to the same side of the spine. As
discussed in Case~IV the selector gadgets of two interlaced edges
create two extra crossings if and only if they are assigned to gaps
on the same side of the x-axis. If we create a bijection that maps a
selector backbone placed in the upper gap to an edge drawn above the
spine, and a selector backbone in the lower gap to an edge drawn
below the spine, we see that an edge crossing on the same side of
the spine in a drawing of $G$ corresponds to two extra crossings in
a labeling of $I_G$ and vice versa. So if $Z$ is the minimum number
of crossings in a spine drawing of $G$, then $2Z$ is the minimum
number of variable crossings in a labeling of $I_G$. Setting
$X=X_1+X_2+X_3+X_4$ this proves Claim~\ref{clm:crossings}. \qed
\end{proof}

It turns out that almost all crossings are fixed (yielding the
number $X$), except for those of selector backbones with vertical
selector segments for which the two underlying edges $(v_i,v_j)$ and
$(v_k,v_l)$ with $i<k$ are \emph{interlaced}, i.e., $i<k<j<l$ holds.
Each pair of interlaced edges creates two crossings in the reduction
if they are assigned to the same side of row $0$ and no crossing
otherwise (see for example edges $(v_1,v_4)$ and $(v_3,v_8)$ or
$(v_2,v_5)$ in Fig.~\ref{fig:4-edge-gadget}). This adds up to at
least $2Z$ crossings, which shows the claim and concludes the
hardness reduction. Furthermore, we can guess an order of the
backbones and then apply
Theorem~\ref{theorem:min-crossings-fixed-order-finite-backbones} to
compute the minimum crossing number for this order.
This shows NP-completeness. \qed
\end{proof}

\paragraph{Acknowledgements.}
This work was started at the Bertinoro Workshop on Graph Drawing
2013. M.N. received financial support by the Concept for the Future
of KIT.  The work of M.~A.~Bekos and A.~Symvonis is implemented
within the framework of the Action ``Supporting Postdoctoral
Researchers'' of the Operational Program ``Education and Lifelong
Learning'' (Action's Beneficiary: General Secretariat for Research
and Technology), and is co-financed by the European Social Fund
(ESF) and the Greek State.  We also acknowledge partial support by
GRADR -- EUROGIGA project no. 10-EuroGIGA-OP-003.

{
\bibliographystyle{abbrv}
\bibliography{mtobl}

\begin{thebibliography}{10}

\bibitem{BekosKSW2007}
M.~A. Bekos, M.~Kaufmann, A.~Symvonis, and A.~Wolff.
\newblock Boundary labeling: Models and efficient algorithms for rectangular
  maps.
\newblock {\em Computational Geometry}, 36(3):215 -- 236, 2007.

\bibitem{FormannW1991}
M.~Formann and F.~Wagner.
\newblock A packing problem with applications to lettering of maps.
\newblock In {\em Proc. Symp. Comput. Geom. (SCG'91)}, pages 281--288. ACM,
  1991.

\bibitem{Hassin1991395}
R.~Hassin and A.~Tamir.
\newblock Improved complexity bounds for location problems on the real line.
\newblock {\em Operations Research Letters}, 10(7):395 -- 402, 1991.

\bibitem{Kau09}
M.~Kaufmann.
\newblock On map labeling with leaders.
\newblock In S.~Albers, H.~Alt, and S.~N{\"a}her, editors, {\em Festschrift
  Mehlhorn}, volume 5760 of {\em LNCS}, pages 290--304. Springer-Verlag, 2009.

\bibitem{k-hmap-55}
H.~W. Kuhn.
\newblock The {H}ungarian method for the assignment problem.
\newblock {\em Naval Research Logistics Quarterly}, 2(1--2):83--97, 1955.

\bibitem{Lin2010}
C.-C. Lin.
\newblock Crossing-free many-to-one boundary labeling with hyperleaders.
\newblock In {\em Proc. {IEEE} Pacific Visualization Symp. (PacificVis'10)},
  pages 185--192. IEEE, 2010.

\bibitem{LinKY08}
C.-C. Lin, H.-J. Kao, and H.-C. Yen.
\newblock Many-to-one boundary labeling.
\newblock {\em Journal of Graph Algorithms and Applications}, 12(3):319--356,
  2008.

\bibitem{mnkf-cmleg-90}
S.~Masuda, K.~Nakajima, T.~Kashiwabara, and T.~Fujisawa.
\newblock Crossing minimization in linear embeddings of graphs.
\newblock {\em {IEEE} Trans. Computers}, 39(1):124--127, 1990.

\bibitem{Neyer2001}
G.~Neyer.
\newblock Map labeling with application to graph drawing.
\newblock In M.~Kaufmann and D.~Wagner, editors, {\em Drawing Graphs}, volume
  2025 of {\em LNCS}, pages 247--273. Springer, 2001.

\bibitem{WS96}
A.~Wolff and T.~Strijk.
\newblock The {M}ap-{L}abeling {B}ibliography (1996). {O}n line:
  http://i11www.ira.uka.de/map-labeling/bibliography.

\end{thebibliography}
}

\end{document}